\documentclass[conference, twocolumn, anonymous, 10pt, a4paper]{IEEEtran} 
\usepackage[utf8]{inputenc}

\usepackage{hyperref}
\usepackage{amssymb}
\usepackage{amsfonts}
\usepackage{amsmath}
\usepackage{bm}
\usepackage{mathtools}
\usepackage{amsthm}
\usepackage{tikz}
\usepackage{pgfplots}
\pgfplotsset{compat=1.13}
\usepackage{xcolor}
\usepackage{color}

\usepackage{url}
\usepackage{xspace}
\usepackage{comment}
\usepackage{cleveref}
\usepackage{ifthen}
\usepackage{xargs}
\usepackage{breakcites}	% allow a box of several citations to
\allowdisplaybreaks
\usepackage{mdframed}
\usepackage{listings} %% Should make our own
\definecolor{codegreen}{rgb}{0,0.6,0}
\definecolor{codegray}{rgb}{0.5,0.5,0.5}
\definecolor{codepurple}{rgb}{0.58,0,0.82}
\definecolor{backcolour}{rgb}{0.95,0.95,0.92}
\definecolor{dkblue}{rgb}{0,0.1,0.5}
\definecolor{lightblue}{rgb}{0,0.5,0.5}
\definecolor{dkgreen}{rgb}{0,0.4,0}
\definecolor{dk2green}{rgb}{0.4,0,0}
\definecolor{dkviolet}{rgb}{0.6,0,0.8}
\definecolor{darkblue}{rgb}{0.0,0.0,0.6}
\definecolor{grey}{rgb}{0.5,0.5,0.5}

\newtheorem{theorem}{Theorem}%[section]
\newtheorem{definition}{Definition}%[theorem]

\newcommand{\from}{\leftarrow}
\newcommand{\rfrom}{\leftarrow_\$}

\newcommandx{\yaHelper}[2][1=\empty]{%
\ifthenelse{\equal{#1}{\empty}}%
  { \ensuremath{ \scriptstyle{ #2 } } } % no offset
  { \raisebox{#1}[0pt][0pt]{\ensuremath{ \scriptstyle{#2}}}}  % with offset
}   
\newcommandx{\yleftarrow}[4][1=\empty, 2=\empty, 4=\empty, usedefault=@]{%
  \ifthenelse{\equal{#2}{\empty}}
  { \xleftarrow{\protect{\yaHelper[#4]{#3}}}} % there's no text below
  { \xleftarrow[\protect{\yaHelper[#2]{#1}}]{ \protect{\yaHelper[#4]{#3}}}} % there's text below
}

\newcommand{\algoname}[1]{\textnormal{\textsc{#1}}}

\newcommand{\Share}{\algoname{Share}}
\newcommand{\Reconstruct}{\algoname{Reconstruct}}
\newcommand{\VSS}{\textup{VSS}}

\newcommand{\ASS}{\mathsf{ASS}}
\newcommand{\RepSS}{\mathsf{RepSS}}

\newcommand{\comp}{\circ}
\newcommand{\eqdef}{\stackrel{\textit{\tiny def}}{=}}

\newcommand{\Zp}{\mathbb{Z}_p}
\newcommand{\PP}{\mathbb{P}}

\newcommand{\Prob}[2][]{\PP_{#1}[ #2 ]}

\newcommand{\til}{{},\ldots ,{}}

\newcommand{\pWhile}{\text{pWhile}\xspace}

\newcommand{\Adv}{\mathcal{A}}
\newcommand{\advid}{a}

\newcommand{\Sim}{\mathcal{S}}

\newcommand{\inp}{\mathsf{in}}

\newcommand{\view}{\mathsf{view}}
\newcommand{\outputfct}{\mathsf{output}}

\newcommand{\MPC}{MPC\xspace}

\newcommand{\IDEAL}{\mathsf{IDEAL}}
\newcommand{\REAL}{\mathsf{REAL}}

\newcommand{\EC}{\text{EasyCrypt}\xspace}
\newcommand{\CCrypt}{\text{CertiCrypt}\xspace}

\begin{document}
\title{Computer-aided proofs for multiparty computation with active security}
\author{\IEEEauthorblockN{Helene Haagh\IEEEauthorrefmark{1} Aleksandr Karbyshev\IEEEauthorrefmark{1} Sabine Oechsner\IEEEauthorrefmark{1} Bas Spitters\IEEEauthorrefmark{1} Pierre-Yves Strub\IEEEauthorrefmark{2}}
\IEEEauthorblockA{\IEEEauthorrefmark{1}Aarhus University, DK}
\IEEEauthorblockA{\IEEEauthorrefmark{2}\'E{}cole Polytechnique, F}
}

\maketitle
\begin{abstract}
Secure multi-party computation (\MPC) is a general cryptographic
technique that allows distrusting parties to compute a function of
their individual inputs, while only revealing the output of the
function. It has found applications in areas such as auctioning, email filtering, and secure teleconference.

Given its importance, it is crucial that the protocols are
specified and implemented correctly. In the programming language community it has become good practice to use computer proof assistants to verify correctness proofs. In the field of cryptography, \EC is the state of the art proof assistant.
It provides an embedded language for probabilistic programming, together with a specialized logic, embedded into an ambient general purpose higher-order logic. It allows us to conveniently express cryptographic properties. \EC\ has been used successfully on many applications, including public-key encryption, signatures, garbled circuits and differential privacy.
Here we show for the first time that it can also be used to prove
security of MPC against a \emph{malicious} adversary. 

We formalize additive and replicated secret sharing schemes and apply them to Maurer's \MPC protocol for secure addition and multiplication.
Our method extends to general polynomial functions. 
We follow the insights from \EC that security proofs can be often be reduced to proofs about program equivalence,
a topic that is well understood in the verification of programming languages.
In particular, we show that in the passive case the non-interference-based (NI) definition is equivalent to a standard game-based security definition.
For the active case we provide a new NI definition, which we call \emph{input independence}.
\end{abstract}

\section{Introduction}

The study of multiparty computation started in the 1980s with the work of Yao \cite{Yao86} and Goldreich et al.~\cite{GMW87}. It has since grown increasingly important and is starting to be used in real-life applications such as auctioning~\cite{mpcbook}, email filtering, secure teleconference~\cite{launchbury2014application}.

A widely used technique for constructing \MPC protocols is secret sharing~\cite{S79,B79}, a cryptographic primitive that distributes a secret among several parties by providing each party with a share of the secret. 
The secret can be reconstructed by combining the shares belonging to a qualified subset of the parties (i.e. parties that are allowed to learn the secret), while other subsets of the parties will have no information on the secret (even when combining their shares).

\EC~\cite{Barthe2011,DBLP:conf/fosad/BartheDGKSS13} has been used to verify cryptographic primitives, and more recently to verify protocols by using it's built-in probabilistic While-language.
While cryptography papers usually provide a presentation of the algorithms
in pseudocode, \EC code of such a protocol is usually not
much longer, but has the benefit of being completely precise.
Moreover, a framework is being developed to tie \EC
into a fully verified tool chain to generate verified low level code
from the protocol definition~\cite{almeida2014verified} and thus obtain high assurance cryptography.

A clear motivation for formal verification is given by Bellare and Rogaway~\cite{cryptoeprint:2004:331}: `In
our opinion, many proofs in cryptography have become essentially
unverifiable. Our field may be approaching a crisis of rigor.' An
 example of the usefulness of formal verification is provided by the vulnerabilities in the Dual EC random bit generator, where the correctness proof was flawed~\cite{hales2013nsa}. Any attempt to formalize this argument would have spotted the gap in the proof. In particular, an attacker that chooses the constants used in Dual EC could potentially predict outputs and this way introduce a backdoor into protocols using Dual EC such as TLS~\cite{checkoway2014practical}. The \emph{feasibility} of verification is, for
instance, demonstrated by the subsequent
verification of improved protocols for elliptic curve cryptography~\cite{ZBB16}.
Finally, formal verification is required to obtain the highest assurance level (EAL7) in Common Criteria.

\subsection{Our Contribution} 
\begin{itemize}
    \item We provide security definitions and proofs for the \MPC protocol by Maurer~\cite{Maurer2006370}. This is the first formalized proof for more than two parties and the first formalized proof of a protocol that is \emph{actively} secure. 
    \item We have formalized these proofs in \EC, a tool that has been used for cryptographic primitives, but only recently also for protocols. A precise description of what we have formalized is presented in section~\ref{whatwehavedone}.
\end{itemize}

We split the protocol into three phases: input, computation and output, where the computation phase can potentially consist of an arbitrary combination of additions and multiplications. This standard approach allows us to treat arithmetic circuits. 
Mathematically, one can see arithmetic circuits as a way to represent multi-variate polynomials on the ring%
\footnote{In the cryptographic literature, most MPC protocols rely on computation over a field. In the case of our protocol, however, a ring is sufficient.} 
$\mathbb{Z}_m$, and therefore, to represent any function over $\mathbb{Z}_m$.

We first discuss simulation-based security definitions for
passive and active security of \MPC protocols. The use of simulation-based notions is the standard approach in cryptography for defining security of MPC protocols  We then proceed with new non-interference-based definitions that are tailored specifically to the class of protocols we consider, and relate them to the simulation-based ones. In particular, we prove that they imply the standard simulation-based cryptographic definitions. In the passive case, non-interference is equivalent to the existence of a simulator. Non-interference (NI) is especially suitable for the computer-aided proofs
presented in \EC, where the probabilistic relational Hoare logic presents a solid foundation
for proving non-interference based statements. A simulation-based proof would proceed by considering an equivalence between a program and a simulator, which are structurally different, whereas NI considers two runs of the \emph{same} program. In the active case, this difference is even bigger because in the security definition the simulator does not obtain the protocol output in advance.
The feasibility of using non-interference for cryptography is known and was e.g. studied by Backes and Pfitzmann \cite{BP04}. In the case of \EC, the benefits were emphasized, in a different context, in a work on masking schemes by Barthe et al.~\cite{DBLP:conf/ccs/BartheBDFGSZ16}. NI is a compositional property, which allows us to build modular proofs.

An important motivation for our work is provided by the \EC formalization of Boolean garbled circuits~\cite{almeida2014verified}.
It provides high assurance crypto, a completely verified tool chain starting
from a readable \EC protocol to verified low level code. We hope to profit from the same technology in the future. 
Garbled circuits provide a framework for secure 2-party computation, a technique complementary to the techniques we use here. 

\subsection{Outline} 
\noindent
Section~\ref{MPC} contains background on secure multi-party computation. Section~\ref{NI} contrasts non-interference based definitions with simulation based ones. Section~\ref{Maurer} discusses our modelling of Maurer's description of active security for the addition protocol. Section~\ref{related} discusses related work and Section~\ref{conclusion} concludes.

The sources are available at
\url{http://users-cs.au.dk/spitters/MPCEC/}.
The implementation compiles with the development version of \EC available from \url{https://github.com/EasyCrypt/easycrypt}.

%%% Local Variables:
%%% mode: latex
%%% TeX-master: "main"
%%% End:

\section{Preliminaries}
\noindent
Let $[n]$ denote the set $\{1 \til n\}$. We will use $\mathbf{x}$ to denote a vector $(x_1,\ldots,x_n)$ and $\mathbf{x}_i$ for its $i$th projection. Conveniently, almost all vectors in this paper have the same length.
Let $x \from \mathcal{D}$ denote the sampling of an element $x$ according to some distribution $\mathcal{D}$, and let $x \rfrom S$ denote that $x$ is sampled from the uniform distribution over the finite set $S$.
We fix an integer $m$ and consider the ring $\mathbb{Z}_m$.

We now want to compare the output distributions of two executions of probabilistic algorithms.

\begin{definition}[Perfect indistinguishability]
 Let $U(x)$ and $V(x)$ be the output distribution of probabilistic algorithm $U$ and $V$, respectively, on input x. Then $U$ and $V$ are perfectly indistinguishable (denoted $U \sim^p V$) if for all inputs $x$, $U(x) = V(x)$.
\end{definition}

\subsection{Secret Sharing} \label{sec:prelim:share}
\noindent
Here we recall the definitions of additive, replicated, and verifiable secret sharing schemes from~\cite{Maurer2006370}. 
The definitions build on top of each other: The verifiable secret sharing scheme combines replicated secret sharing with additional communication, and replicated secret sharing uses additive secret sharing internally.
Looking ahead, we want to state the definition of our replicated and verifiable secret sharing scheme for the special secrecy structure%
\footnote{Secrecy structure: the collection of ignorant party subsets, i.e., it contains all subsets of the parties that cannot learn anything about the secret.}
with privacy against a single party%
\footnote{However, this can easily be extended to corruption of more parties by adapting the secret sharing schemes to the corresponding secrecy structure. See \cite{Maurer2006370} for details.}.
I.e. let $\mathcal{P} = \{P_1 \til P_n\}$ be the set of parties, then for each $i \in [n]$ $\{P_i\}$ is in the secrecy structure $\Sigma$.
The replicated secret sharing scheme with this access structure requires the underlying additive secret sharing scheme to provide privacy against any set of $n-1$ colluding parties.
Throughout this work, we will refer to the party that creates a secret sharings as the \emph{dealer} (of the corresponding secret).

\paragraph{Additive Secret Sharing.} 
The additive secret sharing scheme for $n$ parties consist of a pair of algorithms: $$\ASS = (\Share, \Reconstruct),$$ which are defined as follows:
\begin{itemize}
\item[]\textbf{Sharing:} The $\ASS.\Share$ algorithm on input a secret $s \in \Zp$, samples $n-1$ values $a_1 \til a_{n-1}$ and computes $a_n = s - \sum_{i=1}^{n-1} a_i$. Then it outputs the shares $(a_1 \til a_n)$.
\item[]\textbf{Reconstruct:} The $\ASS.\Reconstruct$ algorithm on input a sharing $(a_1 \til a_n)$, computes and outputs $s = \sum_{i=1}^{n} a_i$.
\end{itemize}

\paragraph{Replicated Secret Sharing.} 
The replicated secret sharing scheme for $n$ parties consist of a pair of algorithms: $$\RepSS = (\Share, \Reconstruct),$$ which are defined as follows:
\begin{itemize}
\item[]\textbf{Sharing:} The $\RepSS.\Share$ algorithm on input a secret $s \in \Zp$, runs the additive sharing algorithm $ (a_1 \til a_n) \from \ASS.\Share(s) $. Then it constructs for all $i \in [n]$
$$ r_i = (a_1 \til a_{i-1}, \bot, a_{i+1} \til a_n) $$
Finally it outputs the shares $(r_1 \til r_n)$.
\item[]\textbf{Reconstruct:} The $\RepSS.\Reconstruct$ algorithm takes a sharing $(r_1 \til r_n)$, extract the additive sharing $(a_1 \til a_n)$ and outputs 
$$s = \ASS.\Reconstruct(a_1 \til a_n)$$
\end{itemize}

\paragraph{Verifiable Secret Sharing.} The verifiable secret sharing scheme for $n$ parties consist of two protocols $\VSS = (\Share, \Reconstruct)$, which are defined as follows:
\begin{itemize}
\item[]\textbf{Sharing:} The $\VSS.\Share$ protocol proceeds as follows
    \begin{enumerate}
        \item The dealer shares their secret $s \in \Zp$ using the replicated sharing algorithm $$ (r_1 \til r_n) \from \RepSS.\Share(s) $$
        i.e., party $P_i$ receives the share\\ $r_i = (a_1 \til a_{i-1}, \bot, a_{i+1} \til a_n)$.
        \item For each $i \in [n]$, each pair of parties in $\mathcal{P}\setminus \{P_i\}$ check whether they received the same value for $a_i$.
        If any inconsistency is detected, the players broadcast a complaint.
        \item The dealer broadcasts all the shares for which a majority of parties raised a complaint. The other parties accept these broadcasted values. If the dealer refuses to broadcast any of the requested shares, the protocol aborts.
    \end{enumerate}

\item[]\textbf{Reconstruct:} The $\VSS.\Reconstruct$ proceeds as follows
    \begin{enumerate}
        \item Party $P_i$ knows share $r_i$ for all $i \in [n]$.
        \item All parties send their share to all other parties such
          that each party knows $r_1 \til r_n$.
        \item Now each party obtained $n-1$ copies of $a_i$ (the underlying additive share). Each party locally does a majority vote and takes the value that occurs more than half of the time.
        \item After the majority vote each party knows $a_1 \til a_n$ and can reconstruct the secret using additive secret sharing 
        $$s = \ASS.\Reconstruct(a_1 \til a_n)$$
    \end{enumerate}
\end{itemize}

\subsubsection{Properties}
All three secret sharing schemes in this section have the following properties:

\paragraph{Correctness} 
Let $s \in \Zp$ be a secret. Then $\Reconstruct(\Share(s)) = s$.

\paragraph{Secrecy}
Let $s_1$ and $s_2$ be secret sharings of two secrets:
\begin{align*}
    (a_1^1, \dots, a_n^1) &\from \Share(s_1), \quad
    (a_1^2, \dots, a_n^2) \from \Share(s_2),
\end{align*}
where party $P_i$ knows shares $a_i^1$ and $a_i^2$. Then for all $i$, $a_i^1 \sim^p a_i^2$.

Both additive and replicated secret sharing schemes can be seen as passively secure secret sharing schemes as they only provide correctness guarantees against an honest dealer. Verifiable secret sharings on the other hand ensures that the dealer's sharing is consistent among parties and hence guarantees correctness against a malicious dealer.

\subsubsection{Linearity of the Secret Sharing Schemes}

The three secret sharing schemes in this section are all linear secret sharing schemes. 
Let $s_1$ and $s_2$ be secret sharings of two secrets:
\begin{align*}
    (a_1^1, \dots, a_n^1) &\from \Share(s_1), \quad
    (a_1^2, \dots, a_n^2) \from \Share(s_2),
\end{align*}
where party $P_i$ knows shares $a_i^1$ and $a_i^2$.
The parties can compute the secret sharing of the sum $s_1 + s_2$ by performing a linear operation on the shares that they know. 
E.g. for additive secret sharing this operation is the addition of the known shares:
\begin{align*}
    rs_i = a_i^1 + a_i^2 \qquad \text{for } i\in[n].
\end{align*}
Then $s_1 + s_2 = \Reconstruct(rs_1, \dots, rs_n).$

\section{Multiparty Computation}\label{MPC}
\noindent
In multi-party computation, $n$ parties wish to compute a deterministic function $f$ of their secret inputs.
For the rest of this work, we fix a (deterministic) function $f : X_1 \times \cdots \times X_n \to Y$, and hence assume for simplicity that all parties obtain the same output\footnote{Different outputs for the parties can be obtained by opening outputs only to certain parties.}. The goal is to compute $y = f(\textbf{x})$, while ensuring that the following security requirements are fulfilled:
\begin{itemize}
    \item Correctness: The correct value of the output $y$ is computed.
    \item Privacy: The output $y$ is the only new information that can be derived from the computation.
\end{itemize}
To achieve this goal, the parties use a probabilistic protocol $\pi$ to compute $y$.

\paragraph{The adversarial model.}
When talking about security of a multi-party computation protocol, we need to consider the power of the adversary.
In our setting, we consider static corruption, where the corrupted parties are determined before the protocol execution.
Furthermore, we consider two types of adversarial behavior: \emph{passive} and \emph{active} corruption, which specify the actions the corrupted parties are allowed to take. 
\begin{itemize}
    \item In the passive model, the adversary is assumed to follow the protocol description, but attempts to learn more than the output from the protocol execution, i.e., tries to learn information that should remain private. This is also known as security against semi-honest or honest-but-curious adversaries.
    \item In the active model, the adversary is allowed to deviate from the protocol description arbitrarily. This is also known as security against malicious adversaries.
\end{itemize}
Moreover, we assume that the adversary is \emph{stateful}, i.e. the adversary can maintain an internal state to log the transcript of the protocol etc.

\subsection{The Protocol}
\noindent
We consider the \MPC protocol presented by Maurer~\cite{Maurer2006370} which provides active security under corruption of $t < n/3$ parties. The protocol is based on verifiable secret sharing and computes a public function of the $n$ parties' inputs that is represented as arithmetic circuit.
The protocol consists of three phases: An input phase, where each of the $n$ parties' secret inputs is shared using verifiable secret sharing; the computation phase, where the parties perform computations on the shares; and finally, the output phase, where each party opens their share of the result, thus allowing everyone to locally reconstruct the result. Evaluating an arithmetic circuit requires additions and multiplications. Like in many other contexts, multiplication is significantly harder to achieve than addition. In our case, the linearity of the secret sharing scheme provides a way of locally adding secret sharings, while multiplication requires further communication.
For simplicity, we focus on the case of corruption of a single party.

The protocol assumes synchronous communication and the existence of private authenticated communication channels between parties as well as an authenticated broadcast channel. Even though the broadcast channel can be simulated, we will use this abstraction to ease readability. Furthermore, the adversary is allowed to be rushing, i.e.\ in each communication step, the adversary can see all the communication from all honest parties before sending its own.

We will now describe the phases of the protocol:

\begin{itemize}
\item[]\textbf{Input:} Each party performs a verifiable secret sharing of their secret input (as described in Section~\ref{sec:prelim:share}). These shares are distributed s.t.\ party $P_j$ receives the $j$'th shares of each secret. At the end of the phase, each party is committed to an input, i.e.~ this input is uniquely determined from the other parties' combined shares. The following matrix represents the knowledge of one party. Each row $i$ stands for the share received from party $P_i$.
\begin{align*}
    \begin{bmatrix}
        a_1^1   &   a_2^1   &   \cdots  &   \bot    &   \cdots  &   a_n^1   \\
        a_1^2   &   a_2^2   &   \cdots  &   \bot    &   \cdots  &   a_n^2   \\
        \vdots  &   \vdots  &   \ddots  &   \vdots  &   \ddots  &   \vdots  \\
        a_1^n   &   a_2^n   &   \cdots  &   \bot    &   \cdots  &   a_n^n   
    \end{bmatrix}     
\end{align*}
\item[]\textbf{Computation:} The parties can now compute on shared values: 
\begin{itemize}
    \item[]\textbf{Addition:} Since the secret sharing scheme is linear, the parties can compute a secret sharing of the sum of any number of shared values by locally adding their shares.  E.g.~given $n$ shares of values represented as matrix of shares, party $P_i$ will collapse each column by adding the values:
\begin{align*}
    ss_j = a_j^1 + a_j^2 + \cdots + a_j^n \qquad \text{for } j \in [n] \text{ s.t. } j \neq i
\end{align*}
This provides party $P_i$ with the $i$'th share of the sum of the secret inputs
\begin{align*}
    rs_i = (ss_1, \dots, ss_{i-1}, \bot, ss_{i+1}, \dots, ss_n).
\end{align*}
    \item[]\textbf{Multiplication:}
    In order to multiply two secret shared values, the parties need to communicate and to introduce fresh randomness. Each party $P_i$ knows the $i$-th share of two shared values $a$ and $b$: 
    \begin{align*}
    \begin{bmatrix}
        a_1   &   a_2   &   \cdots  &   \bot    &   \cdots  &   a_n   \\
        b_1   &   b_2   &   \cdots  &   \bot    &   \cdots  &   b_n   \\
    \end{bmatrix} 
    \end{align*}
    Note now that $a\cdot b = \sum_{i,j} a_i b_j$, meaning that if each party knows a secret sharing of each term $a_i b_j$, then they can compute $a\cdot b$ by (locally) adding the sharings of all terms. 
    
    The parties proceed as follows to compute a secret sharing of term $a_i b_j$: First, each of the $n-2$ (or $n-1$ if $i=j$) parties that knows both $a_i$ and $b_j$ will compute a fresh secret sharing of $a_i b_j$. In a next step, the parties check if all sharings are sharings of the same value by first computing the pairwise differences between two sharings and then opening and reconstructing the difference. If all differences are 0, then the parties choose the sharing of an arbitrary party as sharing of $a_i b_j$. If any of the opened differences is different from 0, then the parties will compute a secret sharing of $a_i b_j$ in the following way: Each party that knows $a_i$ reports their value of $a_i$ to all other parties. Each party sets $a_i$ to be the majority over all received values for $a_i$. The parties do the same for $b_j$. Then, each party sets the additive sharing of the term $a_i b_j$ to be
    $(a_i b_j, 0, \dots, 0)$, computes the corresponding replicated secret sharing, and stores its share of it.
    
\end{itemize}
\item[]\textbf{Output:} Each party opens their share $rs_i$ of the result by sending it to all other parties. This allows everyone to locally reconstruct the sum $y$ of the secret inputs:
\begin{align*}
    y = \VSS.\Reconstruct(rs_1 \til rs_n).
\end{align*}
\end{itemize}

For simplicity, we will consider two kinds of protocol, with a single addition and multiplication, resp., as computation phase. They will be referred to as addition and multiplication protocol, respectively.

\subsection{Passive Security}

This section presents a definition of passive security against one corrupt party following the simulation paradigm, the current standard approach for defining security of \MPC protocols~\cite{GMW87}.

In the setting of passive security, the adversary must follow the protocol description (i.e. the adversary must use his designated input and must send the correct messages). This might seem like a very weak security model, since it does not capture even small deviations from the protocol description. However, it guarantees that the protocol does not leak any information inadvertently (i.e. that an honest-but-curious adversary does not learn unwanted information from the transcript of the protocol). Another way of interpreting passive security is that it provides security against corruption after the execution of the protocol, i.e. what honest parties learned and stored during the protocol execution does not leak information.

The goal is to ensure that the adversary only learns the output of the
computation, which means that everything it sees during the execution
can be computed based on its input and the output. This property is
proved by building a simulator. This simulator is given the input and output of the corrupt party and computes a view (or transcript) that is indistinguishable from the adversary's view in the real execution of the protocol.

\begin{definition}[View and output]
Let $\pi$ be a protocol for computing $f$. 
We define the view of a party as all the messages sent and received by this party during the protocol execution, i.e., we can denote the view of party $P_i$ by
\begin{equation*}
    \view_i^\pi(\mathbf{x}) \eqdef (\mathbf{x}_i,m_i^{(1)} \til m_i^{(t)}),
\end{equation*}
where $m_i^{(j)}$ denote the $j$th message sent or received by party $P_i$. 
Furthermore, we denote the common output of all parties by 
\begin{equation*}
    \outputfct^\pi(\mathbf{x}).
\end{equation*}
\end{definition}

Note that we are considering the setting of deterministic functions (like addition and multiplication) and corruption of a single party.
This means that for passive security, we can consider the correctness and privacy separately~\cite[Chapter 2]{HL10-2PC}. 

\begin{definition}[Perfect Passive Simulation-based Security, \cite{HL10-2PC}] \label{def:sim_passive_security}
We say that an $n$-party protocol $\pi$ securely computes $f$ in the presence of static semi-honest adversaries, if 
\begin{itemize}
    \item[]\textbf{Correctness:} For every $\mathbf{x} \in X_1 \times \dots \times X_n$,
    \begin{equation*}
        \Prob {\outputfct^\pi(\mathbf{x}) = f(\mathbf{x})} = 1.
    \end{equation*}
    \item[]\textbf{Privacy:} For all $i \in [n]$, and for all adversaries that passive corrupts party $P_i$, there exists a polynomial-time simulator $\Sim_i$ such that 
    \begin{equation*}
        \left\lbrace \Sim_i \left( \mathbf{x}_i, f(\mathbf{x}) \right) \right\rbrace_{\mathbf{x}} 
            \sim^p
        \left\lbrace \view_i^\pi(\mathbf{x}) \right\rbrace_{\mathbf{x}}
    \end{equation*}
    where $\mathbf{x} \in X_1 \times \dots \times X_n$.   
\end{itemize}
\end{definition}

\subsection{Active Security}\label{MPC:active_sec}

In this section, we discuss the existing cryptographic definition of active security for \MPC based on the ideal-real world model.
We present a natural extension of the two-party definition by Hazay and Lindell~\cite{HL10-2PC} to the case of $n$ parties.%
\footnote{Note that this extension is restricted to our setting of perfect security in the case of one corrupted party. However, this can be extended to corruption of several parties.}

Active security models the setting where a malicious (or actively corrupt) party may follow any strategy (including arbitrarily deviating from the protocol description). Thus, it is insufficient to consider the adversary's view in the protocol based on its input and output (like in the passive case). This is especially important given that the adversary may try to change its input during the protocol execution, make the output be incorrectly distributed, or make the honest parties output different or incorrect values (to mention a few possible strategies). 

To capture these threats, we consider the ideal-real world model, where we compare the real execution of the protocol with an ideal execution that is secure by definition. 
In the ideal execution, the computation is performed by a trusted party $T$. This party is incorruptible and acts as a black-box such that no one can observe or influence the computation performed by this trusted party. 
This means that the only ``attack'' we allow the adversary to perform is input substitution, i.e., the only thing the adversary is able to do is to change its ``given'' input to something else while this new input cannot depend on the inputs of the honest parties (since the adversary in the ideal execution receives no information before having to commit to its input).
Intuitively, security is shown by providing a simulator (with access to the real-world adversary) that when interacting with $T$ in the ideal world will produce the same output distribution as in the real world.
Note that this definition captures both correctness and privacy of a protocol: Privacy follows from the simulation paradigm and correctness from the fact that the ideal model always outputs the correct result.

We remark that we consider the setting of one corrupt party to match the setting used in the rest of the paper. However, the definition can easily be extended to a setting of several corrupted parties. In the information-theoretic setting that we consider in this work, protocols can achieve active security under corruption of at most $t < n/2$ parties assuming broadcast (depending on the protocol) \cite{RB89}. However, Maurer's protocol allows only for $t < n/3$, even with broadcast~\cite{Maurer2006370}.

Like in the setting of passive security, our definitions work for both polynomially bounded and unbounded simulators.

\paragraph{Notation.}
Let $P_1 \til P_n$ be $n$ parties, and let $\Adv$ denote the adversary that decides to corrupt party $P_\advid$ with $\advid \in [n]$.
Let $T$ be a trusted party that on inputs $\mathbf{x}$ correctly computes $f(\mathbf{x})$.

\paragraph{Ideal Model.} In this model each party sends their input to a trusted party $T$ that computes the function $f$ and returns the result to the parties. 
\begin{itemize}
\item[]\textbf{Input:} Let $x_i$ be the input of party $P_i$ for $i \in [n]$.

\item[]\textbf{Send inputs:} The honest parties $P_i$ for $i \in [n] \setminus \{\advid\}$ send their inputs $x_i$ to the trusted party $T$. The corrupt party $P_\advid$ sends its prescribed input $x_\advid$, some other input, or a special abort symbol $\bot$ to the trusted party $T$. This decision is made by the adversary $\Adv$ and may depend on $x_\advid$ and the adversary's internal state.\\
Let $\mathbf{x}'$ denote the inputs that $T$ receives, where $\mathbf{x}'_i = \mathbf{x}_i$ for all $i \neq \advid$.

\item[]\textbf{Receive outputs:} If $\mathbf{x}'_i = \bot$ (the abort symbol) for some $i\in[n]$, then $T$ informs all parties that the protocol aborts by sending $\bot$ to all parties.  
Otherwise, the trusted party $T$ computes $y = f (\mathbf{x}')$, and sends $y$ to all parties.

\item[]\textbf{Output:} The honest parties output $y$, while the adversary $\Adv$ outputs an arbitrary function $g$ of the prescribed input $x_\advid$ of the corrupts party and the value $y$ obtained from the trusted party (i.e $g(x_\advid,y)$).
\end{itemize}

Let $\IDEAL_{f,\Adv,\advid}(\mathbf{x})$ denote the output of the ideal execution. This is an $n$-tuple containing the honest parties' outputs $y$ and the output of the adversary $\Adv$ from the above ideal execution
\begin{equation*}
    \IDEAL_{f,\Adv,\advid}(\mathbf{x}) := \left(y_1 \til y_n \right)
\end{equation*}
where $y_\advid = g(x_\advid,y)$ and $y_i = y$ for $i \in [n] \setminus \{\advid\}$

\paragraph{Real Model.}
A real execution of the protocol $\pi$ (with no trusted party).
In this case, the adversary $\Adv$ sends all messages on behalf of the corrupt party $P_\advid$ and may follow an arbitrary strategy. The honest parties must follow the protocol description. 

Let $\REAL_{\pi,\Adv,\advid}(\mathbf{x})$ denote the output tuple containing the honest parties' outputs and the output of the adversary $\Adv$ from the real execution of $\pi$.

\begin{definition}[Perfect Active Simulation-based Security, \cite{HL10-2PC}]
\label{def:sim_active_security}
Let $\pi$ be a $n$-party protocol that computes $f$. Protocol $\pi$ is said to securely compute $f$ in the presence of static malicious adversaries, if for every adversary $\Adv$ for the real model, there exists an polynomial-time adversary $\Sim$ (called a simulator) for the ideal model, such that for every $\advid \in [n]$
\begin{equation*}
    \left\lbrace \IDEAL_{f,\Sim,\advid}(\mathbf{x}) \right\rbrace_{\mathbf{x}}
    \sim^p
    \left\lbrace \REAL_{\pi,\Adv,\advid}(\mathbf{x}) \right\rbrace_{\mathbf{x}}
\end{equation*}
where $\mathbf{x}_i \in X_i$ for $i \in [n]$.
\end{definition}

\paragraph{Input extraction}

Any protocol satisfying Definition~\ref{def:sim_active_security} must allow the simulator to extract the input of corrupt parties from the messages they send (up to equivalence, i.e.~ inputs that lead to the same output). 
The reason is that in the real world, the simulator can only get the output that an adversary would see from the ideal functionality by using the adversary's input. Hence, in a protocol that outputs the correct result, it must be possible to extract the adversary's input.

%%% Local Variables:
%%% mode: latex
%%% TeX-master: "main"
%%% End:

\section{Modelling privacy as input-independence}\label{NI}
\noindent
In this section, we present new security definitions against both passive and active corruption where privacy is based on input-independence, and prove that these new definitions imply the simulation-based definitions from Section~\ref{MPC}.
These definitions will allow us to prove active security of Maurer's protocol  in EasyCrypt in Section~\ref{Maurer}.

In the following sections, we redefine security using a
non-interference-based strategy instead of a simulation-based definition, i.e., we compare the adversary's view in two different executions of the protocol under conditions that rule out trivial distinguishability. Input-independence then means that the adversary, given only its input and the output of computation, cannot distinguish which of all possible consistent inputs of the honest parties was used.

\subsection{Input-independence}
\noindent
Input independence is the non-interference property of protocols where the view of a party in the protocol is independent of the other parties' inputs. 
In our setting, this property will only hold only under the side condition that the inputs of the other parties must be consistent (i.e.~the two executions of the protocol must lead to the same output).

\begin{definition}[Non-interference]
An $n$-party protocol $\pi$ enjoys non-interference if for all $i\in[n]$, we have that for all inputs $\mathbf{x},\mathbf{x}' \in X_1 \times \cdots \times X_n$ related under some condition $C_i(\mathbf{x},\mathbf{x}')$, then it holds that 
\begin{equation*}
    \left\lbrace  \view_i^\pi(\mathbf{x}) \right\rbrace_{\mathbf{x}} 
        \sim^p
    \left\lbrace  \view_i^\pi(\mathbf{x}') \right\rbrace_{\mathbf{x}'}.
\end{equation*} 
\end{definition}

The definition states that given two sets of inputs that are related under some condition, then party $P_i$ cannot distinguish between the two executions of the protocol (i.e. party $P_i$'s view in the two executions are indistinguishable).

Looking ahead, we will model the view of the adversary in \EC using the global state of the adversary (\lstinline|glob A|). Thus, given an adversary that corrupts party $P_i$, we can express the non-interference property in \EC pseudocode as follows
\begin{center}
\lstinline|equiv [ $\pi$ ~ $\pi$ : ={glob A} /\ $C_i$ ==> ={glob A}]|
\end{center}
The condition \lstinline|={glob A}| means that the global state of the adversary before the two executions of $\pi$ are equal, while after the executions we have equality over the distribution of the adversary's global state.

\subsection{Passive security} 

Passive security will be defined again as two properties of a protocol $\pi$: correctness and privacy. However, we define the privacy property now as an input-independence property.

Recall that $\view_i^\pi(\mathbf{x})$ denotes the view of party $P_i$ in the execution of the protocol on inputs $\mathbf{x} = (x_1 \til x_n)$, and that $\outputfct^\pi(\mathbf{x})$ denotes the common output of all parties.

For the privacy property in this definition, we will fix the input of the corrupt party to be the same in both executions. The inputs of the honest parties are chosen consistently such that the output of the computation is the same in both executions. Then input-independence means that the view of a corrupt party is independent of the actual inputs of the honest parties as long as they lead to the same output.

\begin{definition}[Perfect Passive NI-based Security] \label{def:NI_passive_security}
We say that an $n$-party protocol $\pi$ securely computes $f$ in the presence of static semi-honest adversaries if
\begin{itemize}
    \item[]\textbf{Correctness:} For every $\mathbf{x} \in X_1 \times \dots \times X_n$,
    \begin{equation*}
        \Prob {\outputfct^\pi(\mathbf{x}) = f(\mathbf{x})} = 1.
    \end{equation*}
    
    \item[]\textbf{Privacy:} For all $i\in[n]$, and for all adversaries that passively corrupt party $P_i$, we have that for all inputs $\mathbf{x}$, $\mathbf{x}' \in X_1 \times \cdots \times X_n$ such that $\mathbf{x}_i = \mathbf{x}'_i$ (fixed input for the corrupt party) and $f(\mathbf{x}) = f(\mathbf{x}')$, then it holds that
    \begin{equation*}
        \left\lbrace  \view_i^\pi(\mathbf{x}) \right\rbrace_{\mathbf{x}} 
            \sim^p
        \left\lbrace  \view_i^\pi(\mathbf{x}') \right\rbrace_{\mathbf{x}'}
    \end{equation*} 
\end{itemize}
\end{definition}

\subsubsection*{Equivalence between the definitions.}
We will now prove that the presented non-interference-based definition is equivalent to the simulation-based definition for perfect \emph{passive} security.

\begin{theorem}\label{passive_eq_ni_sim}
Let $f$ be an efficiently invertible function and let $\pi$ be an $n$-party protocol that computes $f$. Then $\pi$ is perfect passive simulation-based secure if and only if $\pi$ is perfect passive NI-based secure.
\end{theorem}

\begin{proof}
Both definitions consist of two parts, correctness and privacy. As correctness is defined the same in both definitions, we only consider the privacy part of the definitions.

Let $\pi$ have perfect passive security under the NI definition. Let $P_i$ be the corrupt party.
Then by definition, there exists inputs $\mathbf{x}, \mathbf{x}' \in X_1 \times \cdots \times X_n$ with $\mathbf{x}_i = \mathbf{x}'_i$ and $f(\mathbf{x}) = f(\mathbf{x}')$, such that the protocol execution on these inputs produce equally distributed views.
Hence, a simulator with input $\mathbf{x}_i$ and $y=f(\mathbf{x})$ can invert $f$ for fixed $\mathbf{x}_i$ to obtain inputs $x'_1 \til x'_{i-1},x'_{i+1} \til x'_n$ (with $x'_j$ possible different from $x_j$ for $j \neq i$) for the honest parties. Then the simulator can construct a view by simulating the protocol $\pi$ on inputs $\mathbf{x}'$ with $\mathbf{x}'_i = \mathbf{x}_i$.

For the other direction, let $\pi$ have perfect passive security under the simulation-based definition.
Let $P_i$ be the corrupt party and let $\mathbf{x}, \mathbf{x}' \in X_1 \times \cdots \times X_n$ such that $\mathbf{x}_i = \mathbf{x}'_i$ and $f(\mathbf{x}) = f(\mathbf{x}')$.
Then by definition, there exists a simulator $S_i$ such that
\begin{equation*}
    \left\lbrace S_i(\mathbf{x}_i, f(\mathbf{x})) \right\rbrace_{\mathbf{x}} 
        \sim^p
    \left\lbrace  \view_i^\pi(\mathbf{x}) \right\rbrace_{\mathbf{x}}
\end{equation*}
and
\begin{equation*}
    \left\lbrace S_i(\mathbf{x}'_i, f(\mathbf{x}')) \right\rbrace_{\mathbf{x}'} 
        \sim^p
    \left\lbrace  \view_i^\pi(\mathbf{x}') \right\rbrace_{\mathbf{x}'}.
\end{equation*}
Since $\mathbf{x}_i = \mathbf{x}'_i$ and $f(\mathbf{x}) = f(\mathbf{x}')$, we have
\begin{equation*}
    \left\lbrace S_i(\mathbf{x}_i, f(\mathbf{x})) \right\rbrace_{\mathbf{x}} 
        \sim^p
    \left\lbrace S_i(\mathbf{x}'_i, f(\mathbf{x}')) \right\rbrace_{\mathbf{x}'} 
\end{equation*}
and hence
\begin{equation*}
    \left\lbrace  \view_i^\pi(\mathbf{x}) \right\rbrace_{\mathbf{x}}
        \sim^p
    \left\lbrace  \view_i^\pi(\mathbf{x}') \right\rbrace_{\mathbf{x}'}
\end{equation*}
by transitivity of $\sim^p$.
\end{proof}

Note that the class of functions $f$ that can be considered in the previous theorem can be extended if one allows computationally unbounded simulation: Given a function $f$, $f(\mathbf{x})$, and $\mathbf{x}_i$, an unbounded simulator can find 
$x'_1 \til x'_{i-1},x'_{i+1} \til x'_n$ (with $x'_j$ possible different from $x_j$ for $j \neq i$)
 such that $f(\mathbf{x}) = f(\mathbf{x}')$.

The \EC\ formulation of NI-based security lemma for the passive case is
{\small
\begin{lstlisting}
equiv [ $\pi$ ~ $\pi$ :
  ={glob A, advid}
  /\ s{1}.[advid{1}] = s{2}.[advid{2}] 
  /\ sum s{1} = sum s{2} 
  ==>
  ={glob A} ]
\end{lstlisting}
}%
\noindent
where \lstinline|N| denotes the number of parties in the protocol.

\subsection{Active Security}\label{NI:active_sec}

In this section, we redefine the definition of active security as three properties that a protocol must follow, and show that these properties imply simulation-based active security.
This is the definition that will be used in Section~\ref{Maurer}.

A simulation-based proof of active security would proceed by considering an equivalence between $\pi$ and a simulator. However, non-interference properties are more amenable to computer-aided proofs, because the two runs of the same program are structurally the same. This is especially important in the active case:
Note that the simulator does not receive the adversary's secret protocol input as input, but must extract it (in our case from the communication between corrupt and honest parties). As a simulator must start to interact consistently with the adversary without knowledge of the inputs of any party or the protocol output,
the protocol and the simulator will have to differ more.

\EC \emph{does} support game hopping proofs using simulators though; see e.g. Barthe et al.~\cite{DBLP:conf/fosad/BartheDGKSS13}. However, our implementation of Maurer's multiplication protocol takes roughly 500LOC, so a game hopping proof would require many large intermediate protocols and hence much code duplication.

Let $\pi$ be an $n$-party protocol that computes $f$.
We assume for the rest of this work that $\pi$ can be split into three phases: input, computation, and output, where the output phase consists of a single round of communication. Let $\pi_1$ denote the combined input and computation phase, and let $\pi_2$ denote the output phase such that $\pi = \pi_2 \circ \pi_1$.
Let $\view_i^{\pi_1}(\mathbf{x})$ denote the view of party $P_i$ in the execution of $\pi_1$, and let $\outputfct^\pi(\mathbf{x}) = (y_1 \til y_n)$ denote the output of all parties after the execution of the protocol $\pi$.

\paragraph*{Input extraction}
We define an extraction function to be used in the security definition: Let $v =\view_i^{\pi_1}(\mathbf{x})$ be the view of party $P_i$ after the execution of $\pi_1$ (the input and computation phase). Then there exists a polynomial-time input extraction function 
$$ \mathbf{x}_i \from \inp_i (v) $$
that takes a view of $P_i$ after execution of $\pi_1$ and outputs party $P_i$'s committed input (i.e.~the input $P_i$ decided to use during $\pi_1$)

Note that the extraction function $\inp_i$ extracts the input that the adversary is committed to after the input phase. The adversary may start with an input and change it during the input phase. After that phase, however, the adversary is committed to an input. In Maurer's protocol, this is the case because of the shares it sent to the honest parties, i.e.~the shares that the honest parties received determine the adversary's input uniquely. 

\begin{definition}[Perfect Active NI-based Security] \label{def:NI_active_security}
Let $\pi = \pi_2 \comp \pi_1$ be a protocol that computes $f$.

Protocol $\pi$ is said to securely compute $f$ in the presence of static malicious adversaries if for all $a\in[n]$ and for every adversary $\Adv$ that actively corrupts party $P_a$, the protocol fulfills the following properties:
\begin{itemize}
    \item[]\textbf{Correctness:} 
    Let $\mathbf{x} \in X_1 \times \cdots \times X_n$ be the inputs to the execution and let $v =\view_a^{\pi_1}(\mathbf{x})$ be the view of corrupt party $P_a$ after the execution of $\pi_1$. Let $\outputfct^\pi(\mathbf{x}) = (y_1 \til y_n)$ be the output of the protocol $\pi$, then for all $i\in[n]$ with $i \neq a$ we have
    \begin{align*}
        \Prob{y_i = f(\mathbf{x}')} = 1 
    \end{align*}
    where $\mathbf{x}'_a = \inp_a(v)$ (the committed input for the corrupt party) and $\mathbf{x}'_j = \mathbf{x}_j$ for all $j \neq a$ (the honest parties inputs).

    \item[]\textbf{Input Independence:} 
    For all inputs $\mathbf{x}, \mathbf{x}' \in X_1 \times \cdots \times X_n$ with $\mathbf{x}_a = \mathbf{x}'_a$ (fixed input for the corrupt party),
    \begin{equation*}
        \left\lbrace \view_a^{\pi_1}(\mathbf{x}) \right\rbrace_{\mathbf{x}}
        \sim^p
        \left\lbrace \view_a^{\pi_1}(\mathbf{x}') \right\rbrace_{\mathbf{x}}.
    \end{equation*}
    
    \item[]\textbf{Output Simulation:} 
    Let $\mathbf{x} \in X_1 \times \cdots \times X_n$ be the inputs to the execution and let $v =\view_a^{\pi_1}(\mathbf{x})$ be the view of party $P_a$ after the execution of $\pi_1$. Then let $y = f(\mathbf{x}')$, where $\mathbf{x}'_a = \inp_a(v)$ and $\mathbf{x}'_i = \mathbf{x}_i$ for all $i \neq a$. We say that the output phase $\pi_2$ preserves privacy if the final messages $\{m_i\}_{i\neq a}$ sent by the honest parties only depend on the view $v$ and the result $y$, and moreover, they can be computed efficiently. I.e.~the final messages follow an efficiently samplable distribution on $v$ and $y$
    \begin{align*}
        \{m_i\}_{i\neq a} \from \mathcal{D}_{v,y}.
    \end{align*}    
\end{itemize}
\end{definition}

\subsubsection*{NI-based implies simulation-based security}
Next, we prove that the above NI-based definition implies the standard simulation-based definition. In fact, the NI-based definition is greatly inspired by a widely used strategy to construct simulators for the specific kind of protocol we have in mind.

The runtime of the simulator that we construct in the proof depends on
the runtime of input extraction and output simulation. Since both are possible in polynomial-time, then the simulator will run in polynomial-time as well.

\begin{theorem}\label{th:actequiv}
Let $\pi = \pi_2 \circ \pi_1$ be a $n$-party protocol computing $f$. If $\pi$ is perfect active NI-based secure, then there exists a simulator $\Sim$ such that $\pi$ is perfect active simulation-based secure.
\end{theorem}

\begin{proof}
We start from a real protocol execution and argue about a simulation strategy. In the simulation-based security definition, there exists a simulator that has to simulate messages sent by the honest parties without knowing their inputs. Moreover, this simulator has oracle access to a trusted party $T$ that knows the honest parties' inputs. 
Therefore, the general simulation strategy is to extract the adversary's committed input from the messages he sends during $\pi_1$. Then the simulator can query the trusted party for the correct output. Given this, the simulator can compute the final messages sent by the honest parties in the output phase.

We have to argue now that this strategy is feasible given a protocol with the properties stated above as well as that the strategy is indistinguishable from a real protocol execution to the adversary. We construct the following simulator:

\paragraph*{Simulator $\Sim$}
\begin{enumerate}
    \item Run $\pi_1$ with the adversary while simulating the honest parties with default inputs (e.g., $\mathbf{x}_i = 0$ for all $i \neq a$). 
    Let $v = \view_a^{\pi_1}(\mathbf{x})$ be the view of the corrupt party $P_a$ after the execution of $\pi_1$ (i.e.,~all communication between the corrupt party and the simulator).
       
    \item Extract the input $\mathbf{x}'_a$ that the adversary is committed to after the input phase as $\mathbf{x}'_a = \inp_a(v)$. Send $\mathbf{x}'_a$ to the trusted party $T$ to obtain the output $y$.
    
    \item Sample the messages that the honest parties send in the output phase as
        $\{m_i\}_{i\neq a} \from \mathcal{D}_{v,y}$. 
\end{enumerate}

We can now show that a protocol execution simulated with $\Sim$ is indistinguishable from a real protocol execution.
Input independence implies that the adversary's view after $\pi_1$ is independent of the honest parties' inputs. In particular, no adversary can distinguish between a view from executing $\pi_1$ with real inputs for the honest parties from one with inputs $0$ for all honest parties. 
Hence, the adversary can not distinguish between the real execution and the simulated execution of $\pi_1$. 
 
Then, the simulator gathers the view of the corrupt party $v = \view_a^{\pi_1}(\mathbf{x})$, and sends $\mathbf{x}'_a = \inp_a(v)$ to the trusted party $T$. The trusted party computes and returns $y= f(\mathbf{x}')$, where $\mathbf{x}'_i = \mathbf{x}_i$ for all $i \neq a$, i.e., the real inputs of the honest parties.
Thus, the simulator gets the correct output because of the correctness property of the protocol. 

Finally, output simulation guarantees that the final messages the honest parties send in the output phase only depend on the view $v$ of the corrupt party after $\pi_1$ and the correct output of the protocol $y$. Thus, the simulator can sample these messages $\{m_i\}_{i\neq a} $ according to the distribution $\mathcal{D}_{v,y}$. 
The adversary knows the final message $m_a$ that it is supposed to send in the output phase since $m_a$ is uniquely determined by its view in the protocol. Thus, given the messages $\{m_i\}_{i \in[n]}$, the adversary can compute the correct output $y$. 
Furthermore, the adversary learns no more than the output, since the messages sent by the simulator (on behalf of the honest parties) only depend on the output and the adversary's view after $\pi_1$ (which did not reveal any information about the honest parties inputs).
\end{proof}

\paragraph*{Remark}
The other direction, from simulation-based security to NI-based security, does not hold in general as the simulator may use a simulation strategy that is incompatible with the NI-definition.

%%% Local Variables:
%%% mode: latex
%%% TeX-master: "main"
%%% End:

\section{Modelling in \EC}\label{Maurer}
\subsection{\EC}
\noindent
EasyCrypt is a proof assistant for verifying the security of
cryptographic constructions in the computational model.
\EC provides
a simple imperative probabilistic programming language \pWhile to
specify protocols. 

As an example of \EC code, consider the additive sharing protocol consisting of share and reconstruct procedures.
{\small
\begin{lstlisting}
proc share_additive(s : zmod) : zmod list = 
{ var mxrd;
  mxrd <$\$$ dlist dzmod (N-1);
  return (s - sum mxrd) :: mxrd;}

proc reconstruct_additive (sx : zmod list) : zmod = 
{ return sum sx; }
\end{lstlisting}}

\noindent
Here
\lstinline|mxrd <$\$$ dlist dzmod (N-1)|
samples from a uniform distribution on lists over \texttt{zmod} of size $N-1$.

Proving is done using a variety of (probabilistic relational) Hoare logics.
Mathematical functions and data types are defined using an ambient
higher order logic and a functional programming language.
\EC has both tactic based interactive proofs, but also automatic
proofs, using an SMT backend.

Modelling and proving is done in two ways. When dealing with
honest parties, we tend to use functional programs, so-called
operators, and use the ambient logic to reason about these programs.
Adversarial code is treated using a module
system and procedure calls. One specifies the module type of the adversarial code, and
proves properties over all possible instances of this module type. 
The module system is connected to the imperative \pWhile\ language%
\footnote{The choice for an \emph{imperative} probabilistic language is not forced. One could also use a functional probabilistic programming language, such
as Rml, instead of pWhile. Rml used in the ALEA Coq-library~\cite{DBLP:journals/scp/AudebaudP09}, the base for
\CCrypt~\cite{DBLP:phd/hal/ZanellaBeguelin10}, the predecessor of \EC. However, such a functional language is not implemented in \EC.}%
, so we reason in
the corresponding Hoare logic.
Usually, the main effort is to find the correct pre- and post-conditions and loop invariants.
Often we are arguing that the adversary is harmless in certain parts
of the protocol. The way to specify this is via the equivalence of adversarial 
(imperative) code and functional code. This kind of reasoning is 
familiar in program correctness. 

\subsection{Modelling the Protocol}

In this section we discuss how to model Maurer's MPC protocol~\cite{Maurer2006370} and prove active security. 

\paragraph{Adversary and phases}
Recall that a malicious adversary can deviate arbitrarily from the protocol, e.g., by sending wrong or malformed messages or aborting the protocol. 
To model these arbitrary actions, we use the abstract module types of \EC to provide an interface to the adversary, while at the same allowing it to deviate from the protocol description. 
Thus, for each stage in the protocol, whenever we want the adversary to do some computation, send information, or receive information, we call the adversary's abstract procedures. E.g., in the output phase, we send the honest parties' shares of the result
\lstinline|psums| to the adversary and ask the adversary to send his share to all other parties:
\begin{center}
   \lstinline|advc <@ A.bxshareofres(psums);|
\end{center}
Note that according to the protocol description, the adversary is supposed to send its share of the result to all other parties (i.e., everyone should receive the same share). However, the adversary has the power to send different and (possibly) wrong shares to the other parties. We model this by letting the adversary return a matrix \lstinline|advc|, where row $i$ is the share that adversary sends to party $P_i$.

In this setting, we can present the general structure of the three phases of the addition protocol in \EC code. Here \lstinline|<@| denote a procedure call.

{\small
\begin{lstlisting}
proc input(s : zmod list) : zmod matrix list = {
  var shares;
  
  (* Distribution of shares: *)
  shares <@ do_sharing(s);
  pshares <- distribute_shares shares;
  
  (* - adversary receives shares from the honest parties *)
  A.recv_shares(pshares.[advid]);

  (* Consistency check: *)
  (* - collect the complaints *)
  rx <@ verify_shares();
  
  (* - adversary logs the requests since they are public *)
  A.recv_rx(rx);

  (* - reply the complaints *)
  bx <@ broadcast_shares(rx);
  
  (* - adversary logs the broadcast values *)
  A.recv_bx(bx);

  (* - fix the parties' views from the broadcast values *)
  pshares <- fix advid pshares bx;

  return pshares;}

proc computation (pshares : zmod matrix list) : zmod matrix = {
  (* - notify the adversary about the start of the phase *)
  A.localsum();
  (* - perform the local addition for the honest parties *)
  return mklocalsums advid pshares; }

proc output (psums : zmod matrix) : zmod list = {
  var advc, advres, resshares;
  
  (* - the adversary receives shares of the result
   *   from the honest parties and sends his own share *)
  advc <@ A.bxshareofres(psums);
  resshares <-distribute_resshares advid advc psums;
  
  (* - get the adversary's result.. *)
  advres <@ A.getres();
  
  (* ... and return results of all parties *)
  return reconstruct_vss advid advres resshares; }
\end{lstlisting}}

The input and output phase for the multiplication protocol are essentially the same.
The computation phase looks as follows. Importantly, the computation phase contains (a lot of) communication which is modelled by matrices keeping track of all messages that were sent and received.

{\small\begin{lstlisting}
  proc multiplication (a, b : zmod matrix) : zmod matrix = {
    var known_shared_terms, shared_terms_rep, sharedterms_rep_distr, opened_diff, sharedterms;

    (* distribution of shared terms $a_ib_j$ *)
    shared_terms_rep <@ mult_term_sharing(a,b);
    sharedterms_rep_distr <- distribute_shared_terms shared_terms_rep;

    (* compute and open pairwise differences of term sharings *)
    opened_diff <@ mult_check_term_sharing (sharedterms_rep_distr);

    (* choose sharing for each term *)
    sharedterms <@ mult_determine_term_sharing (a,b,sharedterms_rep_distr, opened_diff);

    (* add all terms: *)
    (* - rearrange for convenience *)
    known_shared_terms <- rearrange_sharedterms sharedterms;
    (* - inform adversary of local computation *)
    A.localmultsum(); 
    (* - add term sharings *)
    return mklocalsums_M ((N-1)*N) known_shared_terms;}
\end{lstlisting}}

\noindent Here \lstinline|advid| is the id of the party corrupted by the adversary.

\paragraph{Communication}
\EC\ has no native support for communication. However, the logic and the module system are rich enough to express this. We use lists and matrices to keep track of the messages that are sent. To model the communication with the adversary, we specify abstract procedures to both send and receive messages. The \EC logic keeps track of the global state of the adversary. In the example above, \lstinline|pshares| is the knowledge that each party has. The various calls to \lstinline|A| are used to send/receive information to/from the adversary, using \EC's stateful modules.

\paragraph{Notes on the implementation}
We note that this implementation has a few limitations regarding the power of the malicious adversary compared to the more general description by Maurer. 

We consider a setting where the adversary can only abort during the input phase. Since we consider a protocol that is secure in the presence of an honest majority, the information shared by the honest parties after a successful input phase will be enough to compute the result of the protocol. 
This means that if the adversary aborts during the output phase, then the honest parties will still be able to reconstruct the result of the computation. 

Another limitation is that in the input phase, the adversary is forced to send shares of its secret before it can see the shares of the honest parties secret. This might seem to limit the adversary's power of choosing its input based on this extra information. However, the adversary is only committed to an input at the end of the phase. Since the adversary's initial shares cannot be forced to be consistent, during the consistency check, the adversary can still change the shares of its input, which will allow to change its input to something else (that possibly depends on the information received so far).

\paragraph{Extraction}

As mentioned in Section \ref{MPC:active_sec} and \ref{NI:active_sec}, a proof of active security requires the extraction of information from the adversary's communication. In particular, we can extract the adversary's input and share of the output it is supposed to have from the messages that he sends and receives in the input phase.

The input extractor \lstinline|extract| collects all messages that the adversary sent to honest parties%
(in our case, all other parties)
during the input phase, and uses their shares as shares of the adversary's input. These shares can now be used to reconstruct the input. The verifiable secret sharing guarantees here that all honest parties received consistent shares that when combined reconstruct to the uniquely determined value that the adversary is committed to after the input phase.
{\small
\begin{lstlisting}
op extract (advid : int) (pview : zmod matrix3) = col N advid pview. 
\end{lstlisting}}

\paragraph{High level structure}

The whole protocol is then as follows. \lstinline|pi1| consists of the input and computation phase.
\lstinline|Protocol| consists of \lstinline|pi1| and the output phase. For the purpose of our proofs, the protocol also calls the extractor \lstinline|extract_advinpsx| and then stores an updated input list \lstinline|secrets|.
{\small
\begin{lstlisting}
proc pi1 (s : zmod list) : zmod matrix list * zmod matrix = {
  var rinp, rcomp;
    
  rinp  <@ input (s);
  rcomp <@ computation (rinp);
    
  return (rinp, rcomp); }

proc protocol(s : zmod list) : zmod list = {
  var inp, out;

  (inp, comp) <@ pi1 (s);

  (* input extraction *)
  advinpsx <- mkaddshares (extract advid inp);
  secrets <- s.[advid <- reconstruct advinpsx];

  out <@ output (comp);

  return out; }
\end{lstlisting}}

\subsection{Outline of the proof for active security}

In this section, we present the key steps in proving active security of Maurer's protocol. This is done by proving that the protocol fulfills the three properties in Definition~\ref{def:NI_active_security}: correctness, input independence, and output simulation. 
We will present each property in \EC code and provide an overview of the proof. %
We use $N$ in the implementation to denote the number of parties. Furthermore, parties are indexed from 0. The lemmas only hold as long as $N \geq 4$, since the reconstruction procedure of the verifiable secret sharing scheme takes the majority over $N-1$ values which is only well-defined if there are at least 3 values to compare.
\lstinline|advid| will denote the id of the adversary.
The final messages that the honest parties sent in the output phase are their shares of the result, and these messages are denoted by \lstinline|comp|. 

Note that \EC cannot check the runtime of code. However, it is easy to see that input extraction and output simulation run in polynomial-time in our case.

\subsubsection{Correctness}
In order to state correctness, we first need to define the inputs to the protocol. Honest parties will use the inputs they were assigned to in the beginning. As mentioned before, the adversary may change its mind about its input, but only during the input phase. Afterwards, it is committed to a unique input. Therefore, we extracted the shares \lstinline|advinpsx| of the input in \lstinline|protocol|, and then define correctness with respect to the updated input list \lstinline|secrets| they define. 
{\small
\begin{lstlisting}
lemma correctness sx : hoare [ protocol :
  sx = s /\ size s = N /\ 0 <= advid < N 
  ==>
  0 <= id < N /\ id <> advid =>
    res.[id] = f secrets ].
\end{lstlisting}}

For the input and output phase, correctness of the protocol steps can be reduced to correctness of the secret sharing scheme. For the computation phase, we prove that if the inputs to a gate are secret sharings of secrets $x_1, \dots, x_t$ ($t$ depends on the gate), then the output of the gate is a secret sharing of the gate function on inputs $x_1, \dots, x_t$. For addition, this property follows from linearity of the secret sharing scheme.
For multiplication, note that if $(a_0, \dots, a_{N-1})$ and $(b_0, \dots, b_{N-1})$ are the additive secret sharings corresponding to the replicated secret sharings of inputs $a$ and $b$, then $a\cdot b = \sum_{i,j} a_i b_j$. Since the secret sharing scheme is linear, it is sufficient to prove that the secret sharing of any term $a_i b_j$ that the parties agree on is actually a secret sharing of the value $a_i \cdot b_j$.
In the multiplication protocol, all honest parties that know $a_i$ and $b_j$ will output a secret sharing of the value $a_i \cdot b_j$, whereas the adversary (if involved, i.e.\ if it knows both $a_i$ and $b_j$) may output a secret sharing of an arbitrary value. The subsequent check computes the pairwise differences between those sharings. If all opened differences are 0, then the presence of at least one honest party guarantees that all secret sharings are sharings of the correct value $a_i \cdot b_j$. Otherwise, the replicated secret sharing corresponding to the additive sharing $(a_i b_j, 0, \dots, 0)$ is a secret sharing of $a_i \cdot b_j$.

\subsubsection{Input Independence} 
After the execution of $\pi_1$ (i.e. the input and computation phases), we show that the knowledge (global state) of the adversary is independent of the inputs of the honest parties, meaning that the adversary cannot distinguish between two executions of $\pi_1$ with different inputs for the honest parties. 

{\small
\begin{lstlisting}
lemma input_independence : equiv [ pi1 ~ pi1 :
     ={glob A, advid}
  /\ 0 <= advid{1} < N
  /\ s{1}.[advid{1}] = s{2}.[advid{2}]
     ==> ={glob A} ].
\end{lstlisting}}

For the input phase, we prove input independence, which follows from secrecy of the secret sharing scheme, and integrity of the output, i.e. each party knows the same additive shares (except for the one they are not supposed to know) of each shared input. Addition does not involve any communication and hence preserves input independence. For multiplication, we will consider again a term sharing $a_i b_j$. If this term sharing has input independence (with respect to the honest parties' inputs into the protocol), then the sum of all term sharings will have this property as well, and so will the output of a multiplication gate.
Having the parties output sharings of $a_i b_j$ preserves input independence because of secrecy of the secret sharing scheme. Furthermore, all honest parties that share a value for the current term will share the same value. When computing and opening the pairwise differences of the sharings, each difference will be 0 if and only if the sharings that were compared are sharings of the same value. In particular, if all sharings are sharings of the same value, then all opened differences will be 0 in both executions, and hence won't provide an adversary with any input-dependent knowledge. If not all differences are 0, then one of the parties (the adversary) must have shared an incorrect value, i.e.\ not $a_i b_j$. In this case, the adversary must have caused the opened difference to be different from 0 (and the adversary knows this difference in advance). Since the whole protocol so far was input independent, the opened differences follow the same distribution in both executions, and hence the check preserves input independence. In the final step, the parties agree on a sharing of $a_i b_j$. If all differences were 0, then the parties will use the secret sharings of one of the parties, which preserves input independence by secrecy of the secret sharing scheme. Otherwise, the parties report their values for $a_i$ and $b_j$ and compute a replicated secret sharing from them. Again, this case can only occur if the adversary knows both $a_i$ and $b_j$, and hence this step preserves input independence.

\subsubsection{Output Simulation}
The messages sent by the honest parties in the beginning of the output phase will only reveal the result of the computation. To prove this, we need to show that these messages follow some distribution on the result and the view of the adversary. This is proven by giving the exact function that maps the result and the adversary's view into the messages the honest parties sent. 

Here we note that these final messages are the honest parties' shares of the result, and the view of the adversary contains its share of the result. From the definition of replicated secret sharing, we notice that the adversary's share of the result contains $N-1$ of the $N$ additive shares that sum up to the result. Thus, we can easily reconstruct the missing value, and construct the honest parties replicated shares of the result.

{\small
\begin{lstlisting}
op finalmsg (advid : int) (y : zmod) (pviewadv : zmod list) : zmod matrix =
  mkseq (fun i =>
  mkseq (fun j =>
    if i = advid then witness
    else if i = j then witness
    else if j = advid then y 
         - sum (drop_elem advid pviewadv)
    else pviewadv.[j]) N) N.

lemma output_simulation : hoare [ protocol :
  size s = N /\ 0 <= advid < N 
  ==>  finalmsg advid y advresshares = comp ].
\end{lstlisting}}%
\subsection{The formalization}\label{whatwehavedone}
We have demonstrated how to formalize active security of MPC protocols in \EC. We have done this by an implementation of Maurer's MPC protocol in \EC. For the addition protocol, we have a complete proof for the three properties correctness, output simulation and input independence that are required by our non-interference based security definition. We have extended this to the multiplication protocol by identifying and formalizing all the invariants for both privacy and correctness of the multiplication protocol, and we have proved correctness, output simulation and input independence. We have permitted ourselves the license not to reprove some parts of the multiplication protocol that were very similar to the  addition protocol. The statements and proof structure are very similar. However, refactoring all our code to convince \EC\ of this fact does not seem to provide enough insight to merit this effort. All these places are clearly documented and marked with \lstinline|admit| in the sources.

Our formalization is substantial: as a very rough measure, it consists of approximately 5000LOC, 1800 of which are used for the addition protocol. The code is dense, as it combines the efficient ssreflect language with SMT-calls (for comparison, the easycrypt standard library consists of 18000LOC). 

More generally, we have given a methodology to attack complex simulation-based proofs of protocols involving much communication. We did this by translating simulation-based proofs to NI-arguments and using the \EC module system to model arbitrary adversarial code. Simulation-based proofs are a standard technique in cryptography, but they are difficult to make mathematically precise. A good example is the ongoing effort to make the simulation-based UC framework completely precise, e.g.\ discussed in~\cite{canetti2015simpler}. The framework is in general well understood, implementation in a proof assistant is still lacking; see~\cite{cryptoeprint:2018:306}.

%%% Local Variables:
%%% mode: latex
%%% TeX-master: "main"
%%% End:

\section{Related work}\label{related}

\EC is a specialized proof assistant for security proofs. To our knowledge, \EC is the only tool that currently allows us to conveniently verify MPC protocols in the manner that we did, since it combines a rich ambient logic with an embedded logic for a probabilistic programming language. However, \EC grew out of a Coq library~\cite{DBLP:phd/hal/ZanellaBeguelin10} and, in principle the techniques we present here could also be used in other proof assistants for higher order logic, or type theory, once one defines the programming logic in the ambient language; see e.g.~\cite{jung2017iris} for a framework that supports imperative, but non-probabilistic, program logics in Coq. The general purpose proof assistant Coq has been used to verify crypto protocols in the foundational cryptography framework~\cite{petcher2015foundational} and in the verification of  an OpenSSL implementation of HMAC~\cite{DBLP:conf/uss/BeringerPYA15}. A similar library could be built in F*, as suggested in~\cite{lowstar}. This would have the added benefit of a build-in SMT solver.

Wysteria is a domain
specific language for \MPC. It has been embedded~\cite{rastogi16wysstar} in the $F^*$ programming language/proof assistant. Various protocols have been verified in Wysteria$^*$ to be secure against a \emph{passive}
adversary. However, all the cryptographic primitives are treated axiomatically as $F^*$ does not include probabilistic computation. 
In this paper, we treat all the aspects, consider \emph{active} security and treat the multiplication protocol.

A simulation-based proof for two parties has been formalized in the Isabelle proof assistant~\cite{butler2017simulate}, based on CryptHOL~\cite{CryptHOL}. The logic of Isabelle is similar to \EC{}s ambient logic. However, Isabelle lacks built-in pHoare logics. They use a shallow embedding of a probabilistic programming language into Isabelle using a monadic interpretation. This is less powerful than the deep embedding used in Certicrypt~\cite{DBLP:phd/hal/ZanellaBeguelin10}, and implicitly in \EC. In Isabelle, they prove security against a passive (semi-honest) adversary of a two-party multiplication protocol using simulation based proofs. In contrast, we prove security of a much more complicated protocol that is secure against an active adversary and works for $n$-parties. This requires us to model the adversary abstractly using \EC{}s modules system, and we thus have the harder job to reason about imperative code. Importantly, we also provide new proof techniques that are more amenable to automation, as they are close to the proof techniques used for program logics.

CryptoVerif~\cite{BlanchetFnTPS16} is an automatic protocol verifier. It targets different goals than \EC\ and is complementary to it. Indeed, CryptoVerif aims to automate cryptographic game transformations. It applies a collection
of game transformations, using a full automatic proof strategy that can be driven by users' hints. On the other hand, \EC, relies on an embedding of Probabilistic Relational Hoare Logic that subsumes cryptographic games transformations.  Its logic being relatively complete, it can be used to prove various properties about a large class of cryptographic
primitives. Although we have not tried it, we are very doubtful that Maurer's proof can be reconstructed automatically.

\section{Conclusions and future work}\label{conclusion}
\subsection{Easycrypt} The \EC logic and module system were a good fit to express these protocols in a natural way, once we found the right way of modelling it.

In our experience, developing in \EC is fairly pleasant. The combination of interactive theorem proving using the ssreflect language combined with automatic theorem proving (SMT) is very powerful. Unfortunately, at the moment the SMT-solver does not provide information which lemmas were used unlike e.g.~\cite{kaliszyk2015hol}. This information could be used to speed up checking the document (which currently takes a couple of minutes), but also be used to prove similar lemmas. Many uses of SMT could be avoided by using a more expressive type system. In particular, good support for (coersive) subtyping would have been helpful. Since much of the low-level proving is automated, leaving out simple type information will often result in the SMT-solver failing without further information.
In this case, it would be helpful to provide a counterexample to the user. This functionality is provided by a number of SMT-solvers and also by quickcheck.
\EC's pWhile language does not support iterators (for-loops). Most of our constructions are iterations over the list of parties. Our first modelling consisted of growing lists by a while loop. We have found it is more convenient to start the loop with an array with default values and update during the while loop.

Finally, we spent much time on the whiteboard trying to connect the code for modelling communication with our visual representation of matrices. A simple evaluator for functional programs would have been useful.

\subsection{Future work}
We have proved the security against one corrupted party. The same methodology works for more corrupted parties, as we can just give more information to the adversary, using \EC's module system. We would make a predicate \lstinline|Honest| on the interval $[0..N)$. Our current formalization checks if the current party \lstinline|i| equals \lstinline|advid|. Instead we would check \lstinline|Honest i|. Currently only the information of one party is sent to the adversary (an \EC\ module). Instead, we would send the information of all the corrupted parties to the adversary. Like in the current protocol, all the corrupted parties would move last.

It would be interesting to code arithmetical circuits in \EC. This would be an effort similar to the encoding of  multivariate polynomials in Coq (\url{https://github.com/math-comp/multinomials}).

From a higher perspective, it would be very interesting to formally connect our work with an efficient implementation, as is done in high-assurance crypto~\cite{almeida2014verified}.

\subsection{Conclusion}
We have presented new security definitions for active security of MPC protocols, shown that they imply the standard ones and formalized the security proof in \EC. This is the first formalized protocol with a proof of active security and the first one for $n$-parties.

\section*{Acknowledgements} 
\noindent
Gilles Barthe showed us how non-interference can be used in the context of MPC for a passive adversary. Ivan Damg\aa{}rd helped us to understand MPC protocols and their security proofs. In the beginning of the project we profitted from discussions with Aslan Askarov, Michael Nielsen, and Mathias Pedersen. We are grateful to all of them.\\
Helene Haagh and Sabine Oechsner were supported by the European Research Council (ERC) under the European Unions's Horizon 2020 research and innovation programme under grant agreement No 669255 (MPCPRO), the Danish Independent Research Council under Grant-ID DFF-6108-00169 (FoCC), and the European Union's Horizon 2020 research and innovation programme under grant agreement No 731583 (SODA). Bas Spitters was supported by the Guarded Homotopy Type Theory project, funded by the Villum Foundation, project number 12386.

\bibliographystyle{plain}
\bibliography{references}

\end{document}